\documentclass[12pt,a4paper,fleqn]{article}

\usepackage{lscape,exscale,german,amsthm}
\usepackage[intlimits]{amsmath}
\usepackage[english]{babel}
\usepackage{rawfonts}
\usepackage{latexsym}
\usepackage[cp850]{inputenc}
\usepackage{epsfig}
\usepackage{slashbox}
\usepackage{multirow}
\renewcommand{\baselinestretch}{1.5}
\newcommand{\bol}[1]{\mbox{\boldmath$#1$}}

\newcommand{\bSigma}{\bol{\Sigma}}

\newcommand{\bm}{\bol{\mu}}
\newcommand{\tbm}{\tilde{\bol{\mu}}}

\newcommand{\bx}{\mathbf{X}}

\newcommand{\by}{\mathbf{Y}}

\newcommand{\bL}{\mathbf{L}}

\newcommand{\bz}{\mathbf{z}}
\newcommand{\bB}{\mathbf{B}}

\newcommand{\bw}{\mathbf{w}}

\newcommand{\bi}{\mathbf{1}}
\newcommand{\bzero}{\mathbf{0}}
\newcommand{\bI}{\mathbf{I}}

\newcommand{\Var}{\mbox{Var}}

\newcommand{\brx}{\breve{\bx}}
\newcommand{\brm}{\breve{\bol{\mu}}}

\newcommand{\tbF}{\tilde{\mathbf{\Phi}}}
\newcommand{\bF}{\mathbf{\Phi}}
\newcommand{\bn}{\boldsymbol{\nu}}

\newcommand{\eps}{\pmb{\varepsilon}}

\newcommand{\bry}{\breve{\by}}

\newcommand{\tbSigma}{\tilde{\bSigma}}
\newcommand{\bO}{\mathbf{O}}
\newcommand{\brn}{\breve{\bn}}

\newcommand{\tbn}{\tilde{\bn}}
\newcommand{\tbB}{\tilde{\bB}}
\newcommand{\tbb}{\tilde{\mathbf{b}}}

\newtheorem{theorem}{Theorem}

\newtheorem{corollary}{Corollary}
\newfont{\tabfont}{cmr7 at 7pt}

\begin{document}

\begin{center}
\vspace*{2cm} \noindent {\bf \large On the Exact Solution of the Multi-Period Portfolio Choice Problem for an Exponential Utility under Return Predictability}\\
\vspace{1cm} \noindent {\sc  Taras Bodnar$^{a}$, Nestor Parolya$^{a}$ and Wolfgang Schmid$^{a,*}$
\footnote{$^*$ Corresponding author. E-mail address: schmid@euv-frankfurt-o.de} }\\

\vspace{1cm} {\it \footnotesize  $^a$
Department of Statistics, European University Viadrina, PO Box 1786, 15207 Frankfurt (Oder), Germany} \\
\end{center}

\begin{abstract}
In this paper we derive the exact solution of the multi-period portfolio choice problem for an exponential utility function under return predictability. It is assumed that the asset returns depend on predictable variables and that the joint random process of the asset returns and the predictable variables follow a vector autoregressive process. We prove that the optimal portfolio weights depend on the covariance matrices of the next two periods and the conditional mean vector of the next period. The case without predictable variables and the case of independent asset returns are partial cases of our solution. Furthermore, we provide an empirical study where the cumulative empirical distribution function of the investor's wealth is calculated using the exact solution. It is compared with the investment strategy obtained under the additional assumption that the asset returns are independently distributed.
\end{abstract}

\vspace{0.7cm}

%\noindent JEL Classification: G11, C18, C54\\
\newpage

\section{Introduction}
Investment analysis and portfolio choice theory are very important and challenging topics in finance and economics nowadays. Since Harry Markowitz (1952) presented his mean-variance paradigm  portfolio theory has become a fundamental tool for understanding the interactions of systematic risk and reward.

It is well known that the mean-variance optimization problem of Markowitz (1952) is equivalent to the expected exponential utility optimization under the normality assumption (see Merton (1969)). Unfortunately, his approach only gives an answer to the one-period portfolio selection problem in discrete time but it says nothing about the multi-period (long-run) setting. Therefore, it is of importance to investigate the multi-period portfolio optimization problem which is of great relevance for an investor as well.  The multi-period portfolio selection problem has been analyzed for different types of utility functions by, e.g., Mossin (1968), Merton (1969, 1972), Samuelson (1969), Elton and Gruber (1974), Brandt and Santa-Clara (2006), Basak and Chabakauri (2010).

The continuous case has already been solved for many types of utility functions in the one- and multi-period case by Merton (1969). A number of generalizations under weaker assumptions are given, among others, by A\"{i}t-Sahalia et al. (2009) and Skaf and Boyd (2009). Exact solutions in discrete time are even more difficult to obtain in the multi-period case. Mossin (1968) considers the case of one risk and one riskless asset. He derives conditions when the multi-period strategy becomes myopic or partial myopic. Frequently this can be achieved by demanding independent asset returns. However, the assumption of independence is unfortunately not fulfilled in many applications.

For an investor it would be very useful to have a closed-form solution of the discrete multi-period portfolio problem. Moreover, it is desirable that the optimal portfolio weights can easily calculated at each period. Analytical solutions of the multi-period optimal portfolio choice problems are hard to obtain and they are only available for some partial cases which are often derived under very restrictive assumptions on the distribution of the asset returns. For instance, a closed-form solution for the multi-period portfolio choice problem exists for the quadratic utility function under the assumption that the asset returns are independently distributed (see Li and Ng (2000), Leippold et al. (2004)).

In the present paper we consider an investor who invests into $k$ risky assets and one riskless asset with an investment strategy based on the exponential utility function
\begin{equation}\label{exp_utility}
U(W_t)=-e^{-\alpha W_t}\,.
\end{equation}
Here $W_t$ denotes the investor's wealth at period $t$ and $\alpha>0$ stands for the coefficient of absolute risk aversion (ARA), which is a constant over time for the exponential utility (CARA utility). The application of the exponential utility function is more plausible than the use of the quadratic utility since the first one is monotonically decreasing. That is why the exponential utility function is commonly used in portfolio selection theory. Moreover, the optimization of the expected exponential utility function leads to the well known mean-variance utility maximization problem and consequently its solution lays on the mean-variance efficient frontier.

We derive a closed-form solution of the multi-period portfolio choice problem with the exponential utility function (\ref{exp_utility}) under the assumption that the asset returns depend on certain predictable variables. The joint process consists of the asset returns and the predictable variables and it is assumed to follow a vector autoregressive (VAR) process. This approach is very popular in finance and it is often used for modeling the asset returns (see, e.g., Campbell (1991, 1996), Barberis (2000), Brandt (2010)).

The rest of the paper is organized as follows. In Section 2, the main result of the paper is given. In Theorem 1 an analytical expression of the portfolio weights is provided for each period. In Corollary 1, the case without a predictable vector is considered while independent asset returns are treated in Corollary 2. In Section 3 a short empirical study is presented. The performance of the derived strategy is compared with the one for independent asset returns. The comparison is performed in terms of the cumulative empirical distribution function of the investor's wealth at the end of the investment period. We find significant improvements if the dependence structure is taken into account. Section 4 contains a short summary.

\section{Multi-Period Portfolio Problem for an Exponential Utility}
In this section we derive the analytical solution of the multi-period portfolio choice problem for an exponential utility function assuming that the asset returns and the predictable variables follow a VAR(1) process.

There are only a few papers in literature where the exponential utility function is considered in the multi-period discrete time setting. For instance, \c{C}anako$\breve{\text{g}}$lu and \"{O}zekici (2009) solved the portfolio choice problem assuming that the stochastic market follows a discrete time Markov chain and all parameters of the asset returns, i.e., mean vector and covariance matrix, depend only on the current state of the stochastic market and not on the previous states which is equivalent to the assumption of independence in our settings. In the paper of Soyer and Tanyeri (2006) a Bayesian computational approach with the exponential utility was presented. The authors write that the solution of the multi-period portfolio choice problem with the exponential utility under the assumption of normality "\textit{...cannot be evaluated in closed form and the optimal portfolio cannot be obtained analytically}". In this paper, however, we present an exact solution assuming that the asset returns follow a vector autoregressive process with predictable variables.

 Let $\bx_t=\left(X_{t,1},X_{t,2},\ldots,X_{t,k}\right)^{\prime}$ denote the vector of the returns of $k$ risky assets and let $r_{f,t}$ be the return of the riskless asset at time $t$. Let $\bz_t$ be a $p$-dimensional vector of predictable variables. We assume that $\by_t=(\bx_t^\prime, \bz_t^\prime)^\prime$ follows a VAR(1) process given by
\begin{equation}\label{VAR}
\by_t=\tilde{\bn}+\tbF\by_{t-1}+\tilde{\eps}_t\,
\end{equation}
with $\tilde{\eps}_t\sim \mathcal{N}(\bzero, \tilde{\bSigma}(t))$, where $\tilde{\bSigma}(t)$ is a positive definite deterministic matrix function. Let $\mathcal{F}_t$ denote the information set available at time $t$. Then $\by_t|\mathcal{F}_{t-1}\sim \mathcal{N}_{k+p}(\tbm_t, \tbSigma(t))$, i.e., the conditional distribution of $\by_t$ given $\mathcal{F}_{t-1}$ is a $k+p$ dimensional normal distribution with mean vector $\tbm_t=E(\by_t|\mathcal{F}_{t-1})=E_{t-1}(\by_t)$ and covariance matrix $\Var(\by_t|\mathcal{F}_{t-1})=\tbSigma(t)$.

The stochastic model (\ref{VAR}) is described in detail by Campbell et al. (2003) who argued that the application of VAR(1) is not a restrictive assumption because every vector autoregression can be presented as a VAR(1) process through an expansion of the vector of state (predictable) variables. The idea behind this approach is to find a vector of predictable variables $\bz_t$ which is mostly correlated with the asset returns and to build a VAR(1) process with respect to the asset returns $\bx_t$ and the vector of predictable variables $\bz_t$. The choice of $\bz_t$ depends rather on the data and not on the utility function. Possible predictable variables are, e.g., the dividend yield (cf. Campbell at al. (2003)), the term spread (see, e.g, Brandt et al. (2006)) or another asset return.

From (\ref{VAR}) we obtain the following model for $\bx_t$ expressed as
\begin{equation}\label{VAR11}
\bx_t=\bL\tilde{\bn}+\bL\tbF\by_{t-1}+\bL\tilde{\eps}_t=\bn+\bF\by_{t-1}+\eps_{t}~~~\text{with}~~\bL=[\bI_k~ \bO_{k,p}]\,,
\end{equation}
where $\bI_k$ is a $k\times k$ identity matrix and $\bO_{k,p}$ is a $k\times p$ matrix of zeros.
Consequently, $\bx_t|\mathcal{F}_{t-1}\sim \mathcal{N}_k(\bm_t, \bSigma(t))$, where $\bm_t=E(\bx_t|\mathcal{F}_{t-1})=\bn+\bF\by_{t-1}$ and $\bSigma(t)=\Var(\bx_t|\mathcal{F}_{t-1})=\bL\tbSigma(t)\bL^\prime$.

 Let $\bw_t=\left(w_{t,1},w_{t,2},\ldots,w_{t,k}\right)^{\prime}$ denote the vector of the portfolio weights of the $k$ risky assets at period $t$. Then the evolution of the investor's wealth is expressed as
\begin{equation}\label{wealth_Sec3}
W_t=W_{t-1}\left(1+r_{f,t}+\bw_{t-1}^{\prime}(\bx_t-r_{f,t}\bi)\right)=W_{t-1}\left(R_{f,t}+\bw^{\prime}_{t-1}\brx_t\right) \,,
\end{equation}
where $R_{f,t}=1+r_{f,t}$ and $\brx_{t}=\bx_{t}-r_{f,t}\bi$ with $\brm_{t}=E_{t-1}(\brx_{t})=\bn+\bF\by_{t-1}-r_{f,t}\bi$.
The aim of the investor is to maximize the expected utility of the final wealth.

The optimization problem is given by
\begin{equation}\label{OP_Sec3}
V(0,W_0,\mathcal{F}_{0})=\max\limits_{\{\bw_s\}_{s=0}^{T-1}}E_t[U(W_T)]\,
\end{equation}
with the terminal condition
\begin{equation}\label{exponent}
U(W_T)=-\exp(-\alpha W_T)~~~\text{for}~~\alpha>0.
\end{equation}
Following Pennacchi (2008) the optimization problem (\ref{OP_Sec3}) can be solved by applying the following Bellman equation at time point $T-t$
\begin{eqnarray}\label{BE_Sec3}
V(T-t,W_{T-t},\mathcal{F}_{T-t})
&=&\max\limits_{\bw_{T-t}}E_{T-t}
\Big{[}\max\limits_{\{\bw_s\}_{s=T-t+1}^{T-1}}E_{T-t+1}[U(W_T)]\Big{]} \\
&=&\max\limits_{\bw_{T-t}}E_{T-t}\Big{[}V(T-t+1,W_{T-t}\left(r_{f,T-t}+\bw^{*\;\prime}_{T-t+1}\brx_{T-t+1}\right),\mathcal{F}_{T-t+1})\Big{]}\nonumber
\end{eqnarray}
subject to (\ref{exponent}), where $\bw^{*}_{T-t+1}$ are the optimal portfolio weights at period $T-t+1$. Note that in contrast to the static case now the vector of optimal portfolio weights $\bw_{T-t+1}$ is a function of the weights of the next periods, i.e., of $\bw_{T-t+1}, \bw_{T-t+2},\ldots,\bw_{T-1}$, what is the consequence of the backward recursion method (see, e.g. Pennacchi (2008)).

For the period $T-1$ we get
{\small
\begin{eqnarray}\label{VT-1}
&&V(T-1,W_{T-1},\mathcal{F}_{T-1})\nonumber\\
&&=E_{T-1}\left[-\exp(-\alpha W_{T-1}(R_{f,T}+\bw^{\prime}_{T-1}\brx_{T}))\right]\nonumber\\
&&=-\exp(-\alpha W_{T-1}R_{f,T})E_{T-1}[\exp(-\alpha W_{T-1}\bw^\prime_{T-1}\brx_T)]\nonumber\\
&&=\exp(-\alpha W_{T-1}R_{f,T})\left(-\exp\left[-\alpha(W_{T-1}\bw^\prime_{T-1}\brm_T-\frac{\alpha}{2}\bw_{T-1}^\prime\bSigma(T)\bw_{T-1}W_{T-1}^2)\right]\right)\rightarrow\text{max}\,.
\end{eqnarray}
}
The last optimization problem is equivalent to
\begin{equation}\label{meanv}
W_{T-1}\bw^\prime_{T-1}\brm_T-\frac{\alpha}{2}\bw_{T-1}^\prime\bSigma(T)\bw_{T-1}W_{T-1}^2\rightarrow\text{max}~~~\text{over}~~\bw_{T-1}\,.
\end{equation}
Taking the derivative and solving (\ref{meanv}) with respect to $\bw_{T-1}$ we get the classical solution for the period $T-1$
\begin{equation}\label{w_t-1}
\bw^*_{T-1}=\dfrac{1}{\alpha W_{T-1}}\bSigma^{-1}(T)\brm_T=\dfrac{1}{\alpha W_{T-1}}(\bL\tbSigma(T)\bL^\prime)^{-1}(\brn_T+\bF\by_{T-1})~~\text{with}~~\brn_T=\bn-r_{f,T}\bi\,.
\end{equation}
In Theorem 1 the multi-period portfolio weights for all periods from $0$ to $T-1$ are given.

\begin{theorem}
Let $\bx_{\tau}=\left(X_{\tau,1},X_{\tau,2},\ldots,X_{\tau,k}\right)^{\prime}$ be a random return vector of $k$ risky assets. Suppose that $\bx_{\tau}$ and the vector of $p$ predictable variables $\bz_{\tau}$ jointly follow a VAR(1) process as defined in (\ref{VAR}). Let $r_{f,\tau}$ be the return of the riskless asset. Then the optimal multi-period portfolio weights  are given by (\ref{w_t-1}) for period $T-1$,
\begin{equation}\label{weigthsT-2}
\bw_{T-2}=\dfrac{1}{\alpha W_{T-2}R_{f,T}}\left(\bL\tbSigma^{-1}({T-1})\tbm^*_{T-1}-\bL\bF^\prime\bSigma^{-1}(T)(\brn_T+r_{f,T}\bF\bL^\prime\bi)\right)\,,
\end{equation}
and
{\footnotesize
\begin{align}\label{weightsT-t}
&\bw^*_{T-t}=\dfrac{1}{\alpha W_{T-t}\prod\limits_{i=T-t+2}^{T}R_{f,i}}\left(\bL\tbSigma^{-1}(T-t+1)\tbm^*_{T-t+1}-\bL\tbF^\prime\tbSigma^{-1}(T-t+2)(\brn_{T-t+3}^*+r_{f,T-t+2}\tbF\bL^\prime\bi)\right)\\
&~~\text{with}~~\tbm^*_{T-t+1}=\tbm_{T-t+1}-r_{f,T-t+2}\bL^\prime\bi ~~\text{and}~~\brn_{T-t+3}^*=\tbn-r_{f,T-t+3}\bL^\prime\bi\nonumber\,,
\end{align}
for $t=3,\ldots,T$.
}
%where
%\begin{equation}\label{fi}
%\bF^*\bSigma^{*\;-1}({T-t})\brn^*= \left\{
%  \begin{array}{l l}
%    \bF^\prime\bSigma({T-t})^{-1}\brn &\quad\text{for}~~~~t=2\\
%   \tbF^\prime\tbSigma({T-t})^{-1}\brn & \quad \text{for}~~~ t=3,\ldots,T\\
%  \end{array} \right. \,.
%\end{equation}
\end{theorem}

\begin{proof}
The value function at time point $T-2$ is obtained by substituting (\ref{w_t-1}) into (\ref{VT-1})
\begin{eqnarray}\label{VT-2}
&&V(T-2,W_{T-2},\mathcal{F}_{T-2})=-E_{T-2}\left[\exp{\left(-\alpha W_{T-1}R_{f,T}-\frac{1}{2}\breve{s}_T\right)}\right]\nonumber\\
&&=-\exp\left(-\alpha W_{T-2}R_{f,T-1}R_{f,T}\right)E_{T-2}\left[\exp{\left(-\alpha W_{T-2}R_{f,T}\bw^\prime_{T-2}\brx_{T-1}-\frac{1}{2}\breve{s}_T\right)}\right]\,,
\end{eqnarray}
where $\breve{s}_T=\brm_T^\prime\bSigma(T)^{-1}\brm_T$. Second, according to the properties of VAR(1) processes we get that
\begin{equation}\label{sT}
\breve{s}_T=\bry_{T-1}^\prime\bF^\prime\bSigma(T)^{-1}\bF\bry_{T-1}+2\bry_{T-1}^\prime\bF^\prime\bSigma(T)^{-1}\brn_T+(\brn_T+r_{f,T}\bF\bL^\prime\bi)^\prime\bSigma(T)^{-1}(\brn_T+r_{f,T}\bF\bL^\prime\bi)\,,
\end{equation}
with $\bry_t=\by_t-r_{f,t}\bL^\prime\bi$. This is a quadratic form with respect to the conditional normally distributed vector $\by_{T-1}$. Moreover, using $\bL\bL^\prime=\bI_k$,
\begin{equation}\label{bwy}
\bw^\prime_{T-2}\brx_{T-1}=\bw^\prime_{T-2}(\bL\by_{T-1}-r_{f,T}\bi)=\bw^\prime_{T-2}\bL\bry_{T-1}\,
\end{equation}
and (\ref{VT-2}), we get
\begin{align}\label{VT2}
&V(T-2,W_{T-2},\mathcal{F}_{T-2})=-\exp\left(-\alpha W_{T-2}R_{f,T}R_{f,T-1}\right)\nonumber\\
&\times E_{T-2}\left[\exp{\left(-\frac{1}{2}\bry^\prime_{T-1}\bB(T)\bry_{T-1}-\mathbf{b}(\bw_{T-2})^\prime\bry_{T-1}-c\right)}\right]\,,
\end{align}
where $\bB_T=\bF^\prime\bSigma(T)^{-1}\bF$, $\mathbf{b}(\bw_{T-2})=\bF^\prime\bSigma(T)^{-1}\brn_T+\alpha W_{T-2}R_{f,T}\bL^\prime\bw_{T-2}+r_{f,T}\bB(T)\bL^\prime\bi$ and $c=\frac{1}{2}(\brn_T+r_{f,T}\bF\bL^\prime\bi)^\prime\bSigma(T)^{-1}(\brn_T+r_{f,T}\bF\bL^\prime\bi)$.

Following Mathai and Provost (1992, Theorem 3.2a.1) the expectation given in (\ref{VT2}) is the moment generating function of the quadratic form in normal variables at point $-1$. Hence, it holds that
{\begin{align}\label{mgfq}
&V(T-2,W_{T-2},\mathcal{F}_{T-2})=-\exp\left(-\alpha W_{T-2}R_{f,T-1}R_{f,T}\right)|\bI+\bB(T)\tbSigma({T-1})|^{-\frac{1}{2}}\nonumber\\
&\times \exp{}\left[-\frac{1}{2}(\tbm_{T-1}-r_{f,T}\bL^\prime\bi)^\prime\tbSigma(T-1)^{-1}(\tbm_{T-1}-r_{f,T}\bL^\prime\bi)-c\right.\nonumber\\
&+\frac{1}{2}\left(\tbm_{T-1}-r_{f,T}\bL^\prime\bi-\tbSigma({T-1})\mathbf{b}(\bw_{T-2})\right)^\prime \nonumber\\
&\times\left.(\bI+\bB(T)\tbSigma(T-1))^{-1}\tbSigma^{-1}(T-1)\left(\tbm_{T-1}-r_{f,T}\bL^\prime\bi-\tbSigma(T-1)\mathbf{b}(\bw_{T-2})\right)\right]\,,
\end{align}
}
where $E_{T-2}[\bry_{T-1}]=\tbm_{T-1}-r_{f,T}\bL^\prime\bi$.
Thus, the optimization problem $V(T-2,W_{T-2},\mathcal{F}_{T-2})\rightarrow$ max is equivalent to $V^{*}(T-2,W_{T-2},\mathcal{F}_{T-2})\rightarrow$ max, where
{
\begin{align}\label{max}
&V^{*}(T-2,W_{T-2},\mathcal{F}_{T-2})=-\frac{1}{2}(\tbm_{T-1}-r_{f,T}\bL^\prime\bi-\tbSigma({T-1})\mathbf{b}(\bw_{T-2}))^\prime\nonumber\\
&(\bI+\bB(T)\tbSigma({T-1}))^{-1}\tbSigma({T-1})^{-1}\left(\tbm_{T-1}-r_{f,T}\bL^\prime\bi-\tbSigma({T-1})\mathbf{b}(\bw_{T-2})\right)\nonumber\\
&=-\frac{1}{2}(\tbm_{T-1}-r_{f,T}\bL^\prime\bi-\tbSigma({T-1})\mathbf{b}(\bw_{T-2}))^\prime \nonumber\\
&\times\left(\tbSigma({T-1})+\tbSigma({T-1})\bB(T)\tbSigma({T-1})\right)^{-1}(\tbm_{T-1}-r_{f,T}\bL^\prime\bi-\tbSigma({T-1})\mathbf{b}(\bw_{T-2}))\,.
\end{align}
}
Because the matrix $\tbSigma({T-1})+\tbSigma({T-1})\bB(T)\tbSigma({T-1})$ is positive definite the maximum of $V^{*}(T-2,W_{T-2},\mathcal{F}_{T-2})$ is attained at $\bw^*_{T-2}$ for which
\begin{equation}\label{eq}
\tbm_{T-1}-r_{f,T}\bL^\prime\bi=\tbSigma_{T-1}\mathbf{b}(\bw^*_{T-2}).
\end{equation}
%where $\bA=(\bI+\bB(T)\tbSigma({T-1}))^{-1}$ and constant $C$ contains all terms which do not depend on $\bw_{T-2}$.
%Taking the derivative of $V^*$ with respect to $\bw_{T-2}$ and setting it equal to zero we get the first order conditions (FOCs) for optimal portfolio weights at the period $T-2$
%{\footnotesize
%\begin{align}\label{ww}
%&\dfrac{\partial V^*}{\partial \bw_{T-2}}=-\frac{\alpha}{2} W_{T-2}R_{f,T}\left(-\bL\bA^\prime\bL^\prime\brm_{T-1}+\bL\bA^\prime\tbSigma({T-1})\bF^\prime\bSigma(T)^{-1}\brn-\bL\tbSigma(T-1)\bA\tbSigma^{-1}(T-1)\bL^\prime\brm_{T-1}\right.\nonumber\\
%&\left.+\bL\tbSigma({T-1})\bA\bF^\prime\bSigma(T)^{-1}\brn+\alpha W_{T-2}R_{f,T}\left(\bL\tbSigma({T-1})\bA+\bL\bA^\prime\tbSigma({T-1})\right)\bL^\prime\bw_{T-2}\right.\nonumber\\
%&\left.+r_{f,T}(\bL\bA^\prime\tbSigma({T-1})+\bL\tbSigma({T-1})\bA)\bB(T)\bL^\prime\bi\right)=\bzero\,.
%\end{align}
%}
%Thus, rearranging terms in (\ref{ww}) we obtain
%{\footnotesize
%\begin{equation}\label{FOC}
%\left(\bL\bA^\prime\tbSigma({T-1})+\bL\tbSigma(T-1)\bA\right)\left(\tbSigma^{-1}(T-1)\bL^{\prime}\brm_{T-1}-\bF^\prime\bSigma(T)^{-1}\brn-\alpha W_{T-2}R_{f,T}\bL^\prime\bw_{T-2}-r_{f,T}\bB(T)\bL^\prime\bi\right)=\bzero\,.
%\end{equation}
%}
Using that $\bL\bL^\prime=\bI_k$ we obtain
\begin{equation}\label{weigths}
\bw^*_{T-2}=\dfrac{1}{\alpha W_{T-2}R_{f,T}}\left(\bL\tbSigma^{-1}({T-1})(\tbm_{T-1}-r_{f,T}\bL^\prime\bi)-\bL\bF^\prime\bSigma^{-1}(T)(\brn_T+r_{f,T}\bF\bL^\prime\bi)\right)\,.
\end{equation}
Furthermore, the equality (\ref{eq}) leads to
{
\begin{align}\label{Vt2}
&V(T-2,W_{T-2},\mathcal{F}_{T-2})=-|\bI+\bB(T)\tbSigma(T-1)|^{-\frac{1}{2}}\exp{(-c)}\nonumber\\
&\times\exp{\left(-\alpha W_{T-2}R_{f,T-1}R_{f,T}-\frac{1}{2}(\tbm_{T-1}-r_{f,T}\bL^\prime\bi)^{\prime}\tbSigma^{-1}(T-1)(\tbm_{T-1}-r_{f,T}\bL^\prime\bi)\right)}\\
&=-|\bI+\bB(T)\tbSigma(T-1)|^{-\frac{1}{2}}\exp{(-c)}\nonumber\\
&\times\exp{\left(-\frac{1}{2}\bry_{T-2}^\prime\tilde{\bB}(T-1)\bry_{T-2}-\tilde{\mathbf{b}}^\prime(\bw_{T-3})\bry_{T-2}-\tilde{c}\right)}\,,
\end{align}
}
where $\tilde{\bB}(T-1)=\tbF^\prime\tbSigma(T-1)^{-1}\tbF$, $\tilde{\mathbf{b}}(\bw_{T-3})=\tbF^\prime\tbSigma(T-1)^{-1}(\tbn-r_{f,T}\bL^\prime\bi)+\alpha W_{T-3}R_{f,T-1}\bL^\prime\bw_{T-3}+r_{f,T-1}\tilde{\bB}(T-1)\bL^\prime\bi$ and $\tilde{c}=\frac{1}{2}(\tbn-r_{f,T}\bL^\prime\bi+r_{f,T-1}\tbF\bL^\prime\bi)^\prime\tbSigma(T-1)^{-1}(\tbn-r_{f,T}\bL^\prime\bi+r_{f,T-1}\tbF\bL^\prime\bi)$.
Taking the conditional expectation from the value function (\ref{Vt2}) with respect to $\mathcal{F}_{T-3}$ we receive
{
\begin{align}\label{Vt3}
&V(T-3,W_{T-3},\mathcal{F}_{T-3})=-|\bI+\bB(T)\tbSigma(T-1)|^{-\frac{1}{2}}\exp{(-c)}\nonumber\\
&\times E_{T-3}\left(\exp{\left(-\frac{1}{2}\bry_{T-2}^\prime\tilde{\bB}(T-1)\bry_{T-2}-\tilde{\mathbf{b}}^\prime(\bw_{T-3})\bry_{T-1}-\tilde{c}\right)}\right)\,.
\end{align}
}
Consequently, the value function for the period $T-3$ has a similar structure than $V(T-2,W_{T-2},\mathcal{F}_{T-2})$ with the only difference that $\bB(T)$ is replaced by $\tbB(T-1)$, $\mathbf{b}(\bw_{T-2})$ by $\tbb(\bw_{T-3})$ and $c$ by $\tilde{c}$. It follows immediately that the optimal portfolio weights at period $T-3$ are
{\footnotesize
\begin{equation}\label{weightsT-3}
\bw^*_{T-3}=\dfrac{1}{\alpha W_{T-3}R_{f,T}R_{f,T-1}}\left(\bL\tbSigma^{-1}(T-2)(\tbm_{T-2}-r_{f,T-1}\bL^\prime\bi)-\bL\tbF^\prime\tbSigma^{-1}(T-1)(\brn_T^*+r_{f,T-1}\tbF\bL^\prime\bi)\right)\,
\end{equation}
}
with $\brn_T^*=\tbn-r_{f,T}\bL^\prime\bi$.

The last step is to use mathematical induction with basis $T-3$ in order to receive the statement of Theorem 1.
\end{proof}

The results of Theorem 1 show us that the optimal portfolio weights at every period of time except the last one depend on the covariance matrices of the next two periods and the conditional mean vector of the next period. This property turns out to be very useful if we want to calculate the optimal portfolio weights for a real data set.

Note that the case without predictable variables is a special case of Theorem 1. In this case the following expressions are obtained.
\begin{corollary}
Let $\bx_{\tau}=\left(X_{\tau,1},X_{\tau,2},\ldots,X_{\tau,k}\right)^{\prime}$ be a random return vector of $k$ risky assets which follows a VAR(1) process as defined in (\ref{VAR}) but without a vector of predictable variables $z_{\tau}$. Let $r_{f,\tau}$ be the return of the riskless asset. Then the optimal multi-period portfolio weights for period $T-1$ are given by
\begin{equation}\label{ohnez}
\bw^*_{T-1}=\dfrac{1}{\alpha W_{T-1}}\bSigma^{-1}(T)\brm_T=\dfrac{1}{\alpha W_{T-1}}\tbSigma^{-1}(T)(\brn_T+\bF\by_{T-1})~~\text{with}~~\brn_T=\bn-r_{f,T}\bi\,
\end{equation}
and for $t=2,\ldots,T$ by
\begin{equation}\label{weightsT-t}
\bw^*_{T-t}=\dfrac{1}{\alpha W_{T-t}\prod\limits_{i=T-t+2}^{T}R_{f,i}}\left(\bSigma^{-1}(T-t+1)\brm_{T-t+1}-\bF^\prime\bSigma^{-1}(T-t+2)(\brn_{T-t+2}+r_{f,T-t+2}\bF\bi)\right)\,,
\end{equation}
\end{corollary}
\begin{proof}
The results of Corollary 1 are obtained in the same way as the results of Theorem 1 by putting $\bL=\bI_k$ and $\by_t=\bx_t$ $(p=0)$.
\end{proof}

In Corollary 2 the return vectors are assumed to be independent.
\begin{corollary}
Let $\bx_{\tau}=\left(X_{\tau,1},X_{\tau,2},\ldots,X_{\tau,k}\right)^{\prime}$ be a sequence of the independently and identically normally distributed vectors of $k$ risky assets, i.e., $\bx_{\tau}\sim \mathcal{N}(\bm, \bSigma)$. Let $r_{f,\tau}$ be the return of the riskless asset. We assume that $\bSigma$ is positive definite. Then for all $t=1,\ldots,T$ the optimal multi-period portfolio weights for period $T-t$ are given by
\begin{equation}\label{iweightsT-t}
\bw^*_{T-t}=\dfrac{1}{\alpha W_{T-t}\prod\limits_{i=T-t+2}^{T}R_{f,i}}\bSigma^{-1}\brm~~~\text{with}~~\brm=\bm-r_{f,T-t+2}\bi\,.
\end{equation}
\end{corollary}
\begin{proof}
Corollary 2 immediately follows from Corollary 1 putting $\bF=\bzero$, $\bSigma(t)=\bSigma$ and $\bn=\bm$.
\end{proof}
The results of Corollary 2 can be obtained as a partial case of \c{C}anako$\breve{\text{g}}$lu and \"{O}zekici (2009), where the stochastic market was presented by a discrete time Markov chain. In that case the asset returns depend on the present state of the market and not on the previous ones which implies the independence of the asset return over time.

It is noted that the dynamics of the optimal portfolio weights in Corollary 2 is hidden in the coefficient of the absolute risk aversion $\alpha$  which is given by $\alpha_{\tau}=\left(\alpha W_{T-\tau}\prod\limits_{i=T-\tau+2}^{T}R_{f,i}\right)^{-1}$. Moreover, the expressions of the weights themselves are proportional to the weights of the so-called tangency portfolio (cf. Ingersoll (1987, p. 89), Britten-Jones (1999)). Because the
tangency portfolio is usually considered as a market portfolio in the single-period allocation
problem (see, e.g., Britten-Jones (1999)) we treat the weights given in (\ref{iweightsT-t}) as the weights of a benchmark portfolio in our empirical study presented in the next section.

\section{Empirical Study}
In this section we apply the results of Section 2 to real data. In following we consider an investor who invests into an international portfolio. The portfolio consists of the capital market indices of five developed stock markets, namely Belgium, Germany, Japan, the UK, and the USA. We deal with weekly data of the MSCI (Morgan Stanley Capital International) indices for the equity market returns from the January 4,2002 to December 4,2009. A process is fitted to the return series
\begin{equation}\label{VARes}
\bx_t=\boldsymbol{\nu}+\mathbf{\Phi}\by_{t-1}+\eps_t \quad \text{with} \quad \eps_t \sim ii\mathcal{N}(\mathbf{0},\mathbf{\Sigma}_{\varepsilon})\,.
\end{equation}
We get
\begin{equation}\label{VARes_int}
\bol{\nu}=
\left[
\begin{array}{r}
4.83e-04\\
1.20e-03\\
6.74e-04\\
5.54e-04\\
2.79e-05\\
\end{array}\right], \mathbf{\Phi}=\left[
\begin{array}{rrrrr}
0.2011& -0.1592&  0.01892& -0.196& 0.455\\
0.3139& -0.1231& -0.00191& -0.511& 0.434\\
0.0487&  0.0888& -0.12131& -0.224& 0.343\\
0.1829 & -0.0889&  0.00988& -0.441& 0.382\\
 0.0766 & -0.0643& -0.03049& -0.114& 0.133\\
\end{array}\right],\; \text{and} \;
\end{equation}
\begin{equation*}
\mathbf{\Sigma}_{\varepsilon}=
\left[
\begin{array}{rrrrr}
0.0013085186 & 0.0010544496 & 0.0004365753 & 0.0009120373 & 0.0006781289\\
0.0010544496 & 0.0013833540 & 0.0005648237 & 0.0010218539 & 0.0008332314\\
0.0004365753 & 0.0005648237 & 0.0007994341 & 0.0004733366 & 0.0003667012\\
0.0009120373 & 0.0010218539 & 0.0004733366 & 0.0010176793 & 0.0006927251\\
0.0006781289 & 0.0008332314 & 0.0003667012 & 0.0006927251 & 0.0007242233\\
\end{array}\right].
\end{equation*}

It is remarkable that the last column of the matrix $\bol{\Phi}$ has the largest values which indicate on a strong positive correlation between the US market and the other markets. Moreover, it shows that the influence of the US market on the return indices is larger than those of the domestic ones. Following Campbell et al. (2003) we choose the stock index of the US market as a predictable variable $z_t$ in our empirical study.

Next, we calculate the weights of the two multi-period portfolio strategies for an exponential utility function. We want to compare the case of correlated return vectors given in Theorem 1 with the case of independent variables given in Corollary 2 which completely ignores the time dependence structure well documented for real data.

The performance of both strategies is compared with each other via an extensive simulation study based on $10^5$ independent repetitions. The multi-period portfolio strategies are constructed for $T \in \{13,26,52,104\}$ and for the coefficient of relative risk aversion (RRA) $\alpha_r \in \{0.8, 2\}$. The RRA $\alpha_r=\alpha W_0$ is chosen as a constant absolute risk aversion (ARA) $\alpha$ in this study (without loss of generality we put $W_0=1$). In order to compare the performance of these two strategies we determine the empirical cumulative distribution function (ECDF) of the investor's terminal wealth for each strategy.

The obtained results are presented in Figure 1 and 2. If we compare the performance of two portfolio strategies by their ECDFs, we should choose the strategy whose distribution function lies below the other because the probability of getting a larger wealth is larger for the strategy with a stochastically smaller distribution function. The strategy based on the weights given in Theorem 1 is denoted by EXP, while the notation EXP-iid is used for the method with the weights of Corollary 2.

Figure 1 presents the results for a smaller value of the coefficient of the relative risk aversion. We observe that EXP overperforms EXP-iid for all considered investment periods $T$. For instance, for $T=104$ the probability of getting a wealth between $60$ and $80$ is equal to roughly $25\%$ for EXP while it is almost zero for EXP-iid. For a small horizon $T$ there exists a small probability of bankruptcy for both strategies but it differs not significantly. For $T\geq52$ the probability of a loss tends to zero. The EXP and the EXP-iid strategies both improve as $T$ becomes larger what indicates their good performance in the long-run setting.

Similar results are obtained for larger values of the coefficient of the relative risk aversion (see Figure 2). The performance of the EXP strategy is better for all $T$. From the other side, the probability of obtaining a larger value of the wealth is for both strategies smaller in comparison to the results presented in Figure 1.

Using the results of both figures we can conclude that the EXP strategy has a higher performance for all $T$ and risk levels. Of course this is not surprising since more information about the distribution of the asset returns is taken into account. On the other hand, ignoring the time dependence of the asset returns weakens the results with respect to the final wealth but it does not influence the probability of being bankrupt at the end of the investment period. Moreover, it has to be noted that the comparison of the ECDFs of the expected utilities is not relevant in our study because both strategies give the maximum expected utility in most of the cases and do not differ significantly.

\section{Summary}
Although the first formulation of the multi-period portfolio choice problem was already provided by Markowitz (1952), there are only a few results on closed-form solutions available in literature. They are mostly derived under the assumption that the asset returns are independently distributed. Merton (1969) discovered that the maximization of the exponential utility function for normally distributed returns is equivalent to the maximization of the mean-variance utility function. \c{C}anako$\breve{\text{g}}$lu and \"{O}zekici (2009) obtained a closed-form solution for the exponential utility function under the assumption that the asset returns are independent.
In general, the derivation of an analytical solution of the multi-period portfolio choice problem with the exponential utility for discrete time was considered to be very difficult (see, e.g., Soyer and Tanyeri (2006)).

In the present paper we derive an exact solution of the multi-period portfolio selection problem for an exponential utility function which is obtained under the assumption that the asset returns and the vector of predictable variables follow a vector autoregressive process of order 1. Under the assumption of independence the obtained expressions of the weights are proportional to the weights of the tangency portfolio obtained as a solution in the case of a single-period optimization problem. We show that only the coefficient of absolute risk aversion depends on the dynamics of the asset returns in this case. The weights of the optimal portfolio derived without a vector of predictable variables are obtained as a partial case of the suggested general solution. In an empirical study we compare the derived multi-period portfolio strategies for real data taking five developed stock market indices. A very good performance of the general solution is observed which always overperforms the weights derived under the additional assumption that the asset returns are independent.

The obtained results can be further extended by taking into account the uncertainties about the parameters of the data generating process. The analytical expressions of the weights can be used to derive the expected mean vector and the covariance matrix of the estimated weights which provide us the starting point for the detailed analysis of their distributional properties. This problem is not treated in the present paper and it is left for future research.

\begin{landscape}
\renewcommand{\baselinestretch}{1.0}
\begin{figure}
\caption{\footnotesize Empirical distribution function of the final wealth after $T$ periods using the portfolio weights of Theorem 1 (EXP) and the weights of Corollary 2 (EXP-iid) for the process considered in Section 3 ($\alpha=0.8$, $10^5$ repetitions).}
\vspace{-0.9cm}
\begin{center}
\includegraphics[angle=270, width=25cm]{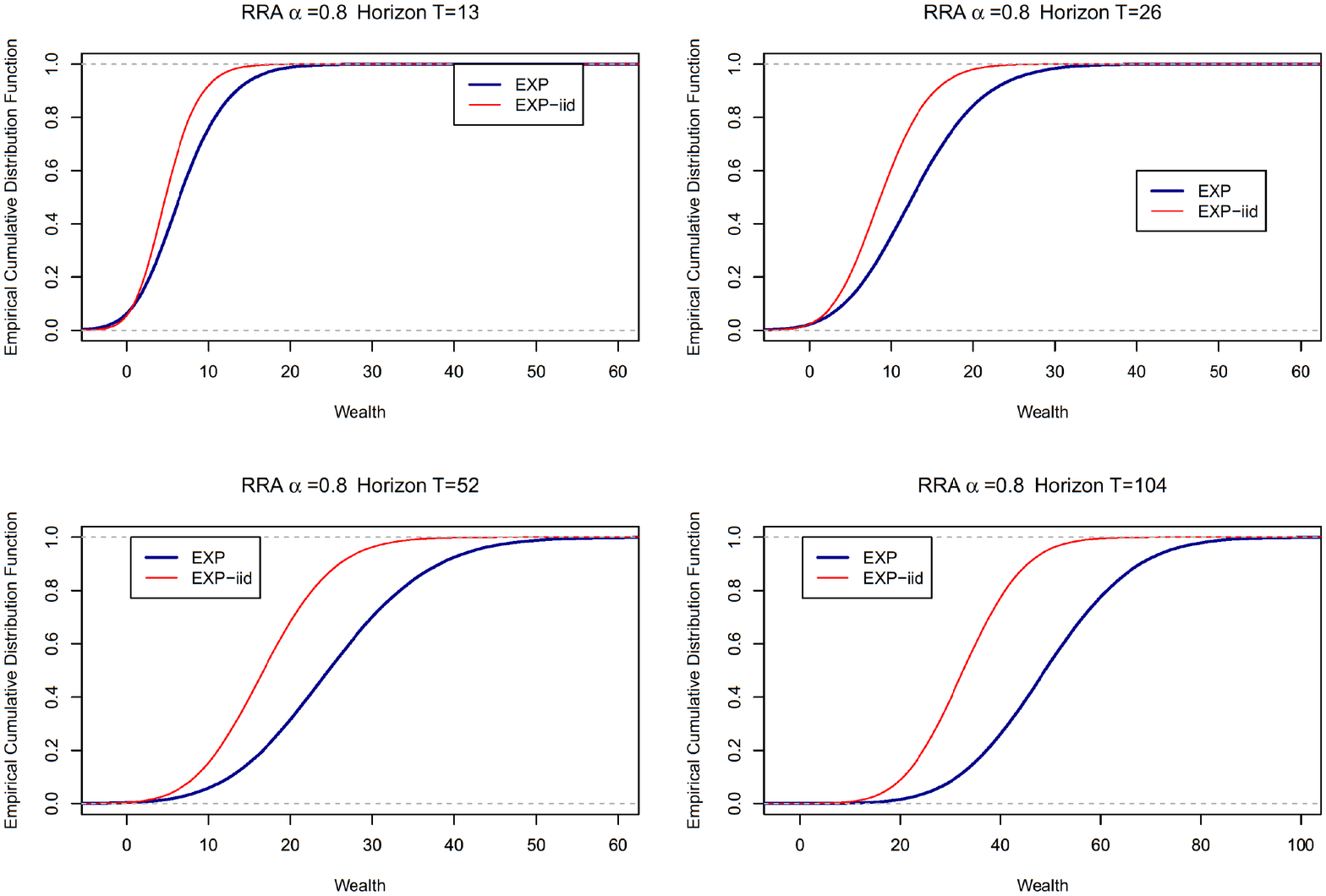}
\end{center}
\label{fig:ECDF_MSCI}
\end{figure}
\end{landscape}

\begin{landscape}
\renewcommand{\baselinestretch}{1.0}
\begin{figure}
\caption{\footnotesize Empirical distribution function of the final wealth after $T$ periods using the portfolio weights of Theorem 1 (EXP) and the weights of Corollary 2 (EXP-iid) for the process considered in Section 3 ($\alpha=2$, $10^5$ repetitions).}
\vspace{-0.9cm}
\begin{center}
\includegraphics[angle=270, width=25cm]{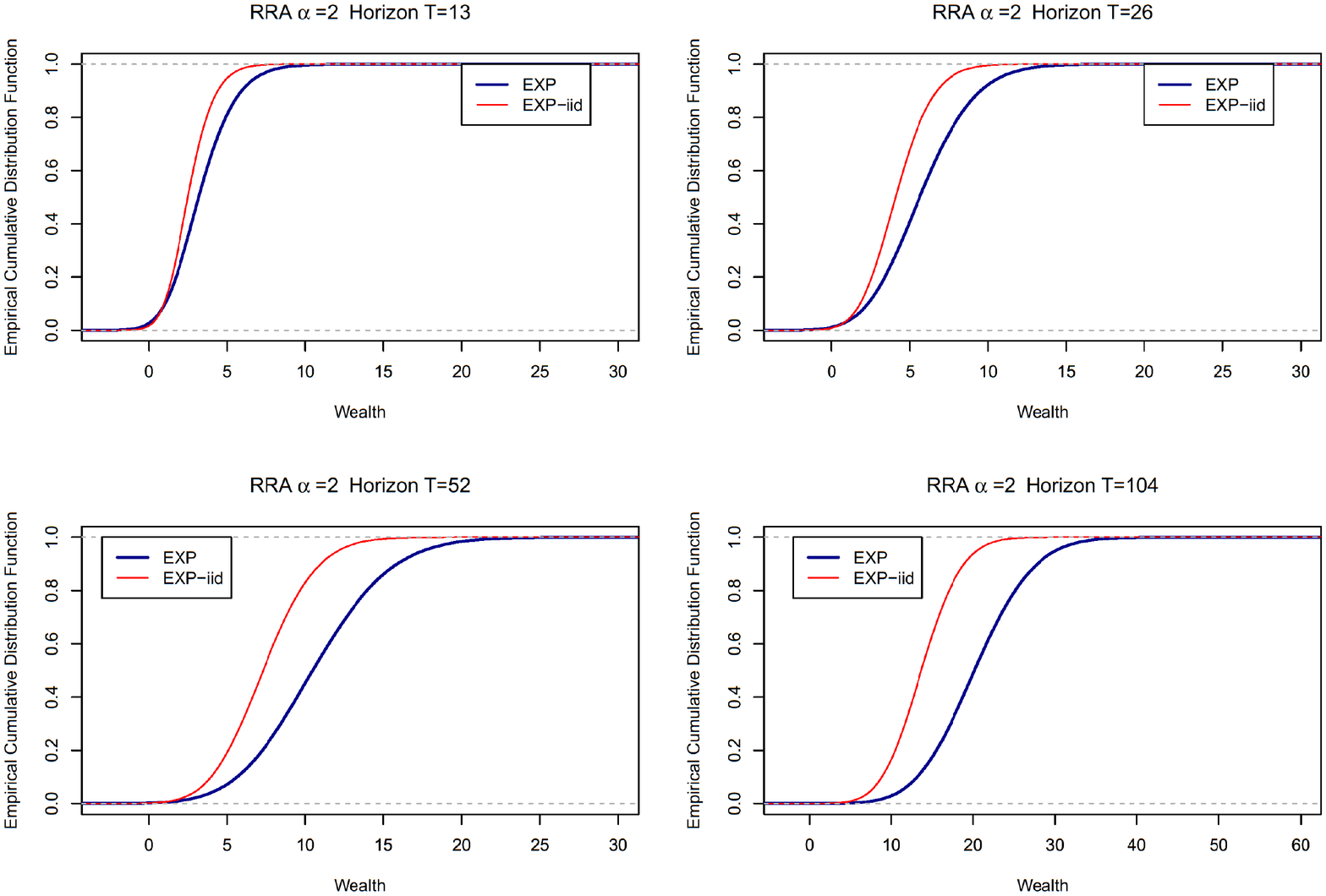}
\end{center}
\label{fig:ECDF_MSCI}
\end{figure}
\end{landscape}


\begin{thebibliography}{99}

\bibitem{} A\"{i}t-Sahalia, Y., J. Cacho-Diaz, T. R. Hurd, (2009), Portfolio choice with jumps: A closed-form
solution. \textit{The Annals of Applied Probability} \textbf{19}, 556-584.

\bibitem{} Arrow, K. J. (1965), Aspects of the theory of risk-bearing. Helsinki: Yrj\"{o} Hahnsson Foundation.

\bibitem{} Basak, S., and G. Chabakauri, (2010), Dynamic mean-variance asset allocation, \textit{Review of
Financial Studies} \textbf{23}, 2970-3016.

\bibitem{} Barberis, Nicholas C., 2000, Investing for the Long Run When Returns Are Predictable,
\textit{Journal of Finance} \textbf{55}, 225-264.

\bibitem{} Brandt, M., and Santa-Clara, (2006), Dynamic portfolio selection by augmenting the asset
space, \textit{The Journal of Finance} \textbf{61}, 2187-2217.

\bibitem{} Brandt, M., Portfolio choice problems, in Y. A\"{i}t-Sahalia and L.P. Hansen (eds.), Handbook
of Financial Econometrics, Volume 1: Tools and Techniques, North Holland, 2010, 269-336.

\bibitem{} Britten-Jones, M. (1999), The sampling error in estimates of mean-variance efficient portfolio
weights, \textit{Journal of Finance} \textbf{54}, 655-671.

\bibitem{} Campbell, J. Y., (1991), A Variance Decomposition for Stock Returns, \textit{Economic Journal}
\textbf{101}, 157-179

\bibitem{} Campbell, J. Y., (1996), Understanding Risk and Return, \textit{Journal of Political Economy}
\textbf{104}, 298-345.

\bibitem{} Campbell, J. Y., Chan Y.L., and Viceira L. M., (2003), A multivariate model of strategic
asset allocation. \textit{Journal of Financial Economics} \textbf{67}, 41-80.

\bibitem{} \c{C}anako$\breve{\text{g}}$lu, E., \"{O}zekici, S., (2009), Portfolio selection in stochastic markets with exponential
utility functions. \textit{Annals of Operations Research} \textbf{166}, 281-297,

\bibitem{} Elton, E. J., Gruber, M. J. (1974), On the optimality of some multiperiod portfolio selection
criteria. \textit{Journal of Business} \textbf{47}, 231-243.

\bibitem{} Ingersoll, J. E. (1987), Theory of Financial Decision Making, Rowman \& Littlefield Publishers.

\bibitem{} Leippold, M., Vanini P. and Trojani F., (2004), A geometric approach to multiperiod
mean-variance optimization of assets and liabilities. \textit{Journal of Economic Dynamics and
Control} \textbf{28}, 1079-1113.

\bibitem{} Li, D., and W. L. Ng, (2000), Optimal dynamic portfolio selection: multiperiod mean-variance
formulation, \textit{Mathematical Finance} \textbf{10}, 387-406.

\bibitem{} Markowitz, H., (1952), Portfolio selection,\textit{ The Journal of Finance} \textbf{7}, 77-91.

\bibitem{} Mathai A. M., Provost S. B., (1992), Quadratic Forms in Random Variables: Theory
and Applications, Marcel Dekker, New York.

\bibitem{} Merton, R. C., (1969), Lifetime Portfolio Selection under Uncertainty: The Continuous
Time Case, \textit{Review of Economics and Statistics} \textbf{50}, 247-257.

\bibitem{} Mossin, J., (1968), Optimal multiperiod portfolio policies, \textit{The Journal of Business} \textbf{41},
215-229.


\bibitem{} Pennacchi, G., (2008), Theory of Asset Pricing, Pearson/Addison-Wesley: Boston.

\bibitem{}  Pratt, J. W., (1964), Risk aversion in the small and in the large, \textit{Econometrica} \textbf{32}, 122-136.

\bibitem{}  Samuelson, P. A., (1969), Lifetime Portfolio Selection By Dynamic Stochastic Programming,
\textit{Review of Economics and Statistics} \textbf{51}, 239-246.

\bibitem{}  Skaf, J., and S. Boyd, (2009), Multi-Period Portfolio Optimization with Constraints and
Transaction Costs. Stanford working paper.

\bibitem{}  Soyer R., and Tanyeri K. (2006), Bayesian portfolio selection with multi-variate random
variance models, \textit{European Journal of Operational Research} \textbf{171}, 977-990.
\end{thebibliography}
\end{document}